%% file: main.tex
\documentclass[11pt]{article} %
\usepackage{pgfplots}
\pgfplotsset{width=7cm,compat=1.9}
\usepackage{authblk}
\usepackage{booktabs}
\usepackage{float}
\usepackage[margin=1in]{geometry}
\usepackage[scr=rsfso]{mathalfa}
\usepackage{mdframed} %
\usepackage{microtype} %
\usepackage{mathpazo}
\usepackage{amsthm}
\usepackage{xspace}
\usepackage{comment}
\usepackage{lipsum}
\usepackage{amsfonts}
\usepackage{graphicx}
\usepackage{epstopdf}
\usepackage{algorithm}
\usepackage[noend]{algorithmic}
\usepackage{tikz}
\usepackage{caption}
\usepackage{amsmath}
\usepackage{hyperref}
\usepackage{amssymb}
\usepackage{dirtytalk}
\usepackage{tcolorbox,thm-restate,multirow}
\usetikzlibrary{decorations.pathmorphing}
\ifpdf
  \DeclareGraphicsExtensions{.eps,.pdf,.png,.jpg}
\else
  \DeclareGraphicsExtensions{.eps}
\fi

\let\epsilon=\varepsilon %

\newcommand{\dist}{{\sf dist}}
\newcommand{\R}{\mathbb{R}}
\newcommand{\B}{\mathcal{B}}
\newcommand{\sol}{\mathsf{sol}}
\newcommand{\np}{\textsf{NP}}
\newcommand{\apx}{\textsf{APX}}

\newcommand\norm[1]{\lVert#1\rVert}

\newtheorem{theorem}{Theorem} 
\newtheorem{lemma}[theorem]{Lemma} %
\newtheorem{proposition}[theorem]{Proposition} %
\newtheorem{definition}[theorem]{Definition} %
\newtheorem{claim}[theorem]{Claim} %

\begin{document}

\title{Parameterized Approximation Algorithms for $k$-Center Clustering and Variants\footnote{A preliminary version appears in AAAI 2022.}}

\author[]{Sayan Bandyapadhyay\thanks{Supported by the European Research Council (ERC) via grant LOPPRE, reference 819416.} \textsuperscript{\rm 1}}

\author[]{Zachary Friggstad\thanks{Supported by an NSERC Discovery Grant and NSERC Discovery Accelerator Supplement Award.} \textsuperscript{\rm 2}}
\author[]{Ramin Mousavi\textsuperscript{\rm 2}}

\affil[]{\textsuperscript{\rm 1} Department of Informatics, University of Bergen, Norway\\
     \textsuperscript{\rm 2} Department of Computing Science, University of Alberta, Canada
\authorcr
  \{\tt sayan.bandyapadhyay@uib.no, zacharyf@ualberta.ca, mousavih@ualberta.ca\}}

\date{}

\maketitle

\begin{abstract}
$k$-center is one of the most popular clustering models. While it admits a simple 2-approximation in polynomial time in general metrics, the Euclidean version is NP-hard to approximate within a factor of 1.93, even in the plane, if one insists the dependence on $k$ in the running time be polynomial. Without this restriction, a classic algorithm yields a $2^{O((k\log k)/{\epsilon})}dn$-time $(1+\epsilon)$-approximation for Euclidean $k$-center, where $d$ is the dimension.

We give a faster algorithm for small dimensions: roughly speaking an $O^*(2^{O((1/\epsilon)^{O(d)} \cdot k^{1-1/d} \cdot \log k)})$-time $(1+\epsilon)$-approximation. In particular, the running time is roughly $O^*(2^{O((1/\epsilon)^{O(1)}\sqrt{k}\log k)})$ in the plane. We complement our algorithmic result with a matching hardness lower bound. 

We also consider a well-studied generalization of $k$-center, called Non-uniform $k$-center (NUkC), where we allow different radii clusters. NUkC is NP-hard to approximate within any factor, even in the Euclidean case. We design a $2^{O(k\log k)}n^2$ time  $3$-approximation for NUkC in general metrics, and a $2^{O((k\log k)/\epsilon)}dn$ time $(1+\epsilon)$-approximation for Euclidean NUkC. The latter time bound matches the bound for $k$-center.  
\end{abstract}

\section{Introduction}
\input{intro}

\section{$k$-center}\label{sec:kcenter}
\input{kcenter}

\section{Hardness}\label{sec:hardness}
\input{hardness}

\section{Non-Uniform $k$-Center}\label{sec:nukc}
\input{nukc}

\bibliographystyle{alpha}
\bibliography{references}

\end{document}

%% file: intro.tex
Given $n$ data points in a metric space, a \emph{clustering} of the data is a partition into a number of groups (or clusters), such that 
points in each group are more similar compared to points across multiple groups. The notion of similarity is captured by the distances between the 
points, which form a metric. The task of partitioning the data as above is referred to also as clustering, which has numerous applications in the areas of artificial 
intelligence (AI), machine learning (ML) and data mining. For example, during the training process of an AI/ML model, the training samples are clustered into similar groups 
to enhance the expressive power of learning methods. With the goal of retrieving the natural clustering of the points, the task of partitioning is modeled as an optimization problem where 
the goal is to find a clustering that optimizes some objective function. For center-based objectives, each cluster is represented by a point, which is called the center 
of the cluster. The similarity (or dissimilarity) of a cluster is measured in terms of the deviation of the points in the cluster from the center. Consequently, this deviation is captured by 
the corresponding objective function, which one needs to minimize. Arguably the most popular center-based objectives are $k$-center, $k$-means, and $k$-median. All of these problems ask to find $k$ center points, 
where $k> 0$ is a given integer. The desired $k$ clusters are formed by assigning each point to its closest center. For $k$-center, the deviation is defined as the maximum distance between a point and its closest center, and for $k$-median, it is the sum of the distances between points and their closest centers. $k$-means is similar to $k$-median except here the deviation is the sum of the square of the distances. 

In this article, we limit our discussions to the Euclidean version of the above clustering problems where the data points and centers belong to a real  
space of dimension $d\ge 2$ and the distance measure is the Euclidean distance. All of these problems are \np-hard even in the plane, i.e., when $d=2$ \cite{DBLP:journals/tcs/MahajanNV12,megiddo1984complexity,feder1988optimal}. For both $k$-median and $k$-means, $(1+\epsilon)$-approximation algorithms (or \emph{approximation schemes}) are known with running time $2^{(f(\epsilon))^{d-1}}n{\log}^{d+6}n$ \cite{kolliopoulos2007nearly} and  $nk{(\log n)^{(d/\epsilon)}}^{O(d)}$ \cite{DBLP:conf/soda/Cohen-Addad18}, respectively, for some function $f$. Note that for constant dimension and in particular, for $d=2$, these algorithms run in polynomial time. In contrast, it is widely known that such an approximation scheme does not exist for $k$-center \cite{mentzer1988approximability,feder1988optimal,chen2021mentzer}. Indeed, it is \np-hard to approximate $k$-center in the plane up to a factor of 1.82 \cite{mentzer1988approximability,chen2021mentzer}. On the positive side, several polynomial-time 2-approximations are known for $k$-center \cite{DBLP:journals/tcs/Gonzalez85,feder1988optimal}. On the hardness side of $k$-median and $k$-means, it is known that polynomial-time approximation scheme is not possible for these problems (or are \apx-hard) when the dimension is not necessarily a constant \cite{awasthi2015hardness,bhattacharya2020hardness}. 

From the above discussion, it is evident that all these problems are intractable when the dimension is arbitrary. To cope with this hardness, researchers have designed \emph{fixed-parameter tractable} (FPT) $(1+\epsilon)$-approximations. Fixed-parameter tractability is a notion in \emph{parameterized complexity}~\cite{cygan2015parameterized} which allows running time to be expressed as a function of a parameter $p$. In particular, an algorithm has FPT running time parameterized by $p$ if it runs in $f(p)\cdot n^{O(1)}$ time, where $f$ is a function that solely depends on $p$. For clustering problems, the most natural parameter is $k$, which is typically small in practice, and hence an FPT $(1+\epsilon)$-approximation algorithm leads to an efficient approximation scheme. Note, in particular, that now the running time does not depend exponentially on $d$. Consequently, Badoiu~et al.~\cite{DBLP:conf/stoc/BadoiuHI02} designed an approximation scheme for $k$-center that runs in time $2^{O((k\log k)/{\epsilon}^2)}dn$, and later Badoiu and Clarkson \cite{DBLP:conf/soda/BadoiuC03} slightly improved the time to $2^{O((k\log k)/{\epsilon})}dn$. For $k$-median and $k$-means, several FPT approximation schemes were designed in a series of work \cite{DBLP:conf/stoc/BadoiuHI02,de2003approximation} culminating in a $2^{(k/\epsilon)^{O(1)}}nd$ time approximation scheme due to Kumar~et al.~\cite{DBLP:journals/jacm/KumarSS10}. 

The study of sub-exponential algorithms is a popular research direction in parameterized complexity. Here the goal is to design a $2^{o(p)}n^{O(1)}$ time algorithm for restricted instances of a problem that admits an $f(p)\cdot n^{O(1)}$ time algorithm \cite{demaine2005subexponential,DornFLRS13}. A central theme of this subarea is to design $2^{O(\sqrt{p}\cdot \text{polylog}(p))}n^{O(1)}$ time algorithms for planar instances of the problems \cite{FominLKPS20,FominLMPPS16,Nederlof20a}, e.g., planar vertex cover \cite{cygan2015parameterized}. ($\text{polylog}(p)$ is a constant power of $\log p$.) This theme is informally known as the \say{square root phenomenon}. 

\renewcommand{\arraystretch}{1}
\begin{center}
\begin{table}[t]
\centering
\begin{tabular}{ |p{3cm}|p{3.1cm}|p{3cm}| }
 \hline
 \centering \textbf{Problem} & \textbf{Algorithms/Upper bounds} & \textbf{Hardness/Lower bounds} \\ 
 \hline
\centering \multirow{2}{*}{$k$-center}  & \centering $2^{O((k\log k)/{\epsilon})}dn$ & 1.82-factor in $n^{O(1)}$\\ 
& \centering $2^{O((1/\epsilon)^{O(d)}k^{1-1/d}\log k)}$ Thm. \ref{thm:kcentercont} & $O^*(2^{2^{\bar\epsilon\cdot\delta\cdot d}\cdot k^{1-\delta}})$ Thm. \ref{hardness: k-center}\\
\hline
\centering \multirow{2}{*}{$k$-median}  & \centering $2^{(f(\epsilon))^{d-1}}n{\log}^{d+6}n$ & \apx-hard \\
& \centering $2^{(k/\epsilon)^{O(1)}}nd$ & \\
\hline
\centering \multirow{2}{*}{$k$-means} & \centering $nk{(\log n)^{(d/\epsilon)}}^{O(d)}$  & \apx-hard\\
& \centering $2^{(k/\epsilon)^{O(1)}}nd$ &\\
\hline
\centering \multirow{2}{*}{NUkC} &  \centering 3-approx. in $2^{O(k\log k)}n^2$ (metric) Thm. \ref{thm:nukcgeneral} & $\gamma$-hard $\forall \gamma$ in $n^{O(1)}$\\
& \centering $2^{O((k\log k)/\epsilon)}dn$ Thm. \ref{thm:nukceuclid} & \\
\hline 
\end{tabular}
\caption{A summary of previous and our work. Our results are marked with theorem numbers.}
\label{table:1}
\end{table}
\end{center}

Motivated by the above research directions, we study sub-exponential algorithms for $k$-center when the dimension $d$ is small. In particular, we ask the following questions. 
\begin{tcolorbox}
Does the planar version of $k$-center ($d=2$) admit a $2^{O(\sqrt{k}\cdot \text{polylog}(k))}n^{O(1)}$ time approximation scheme? Does $k$-center admit a $2^{o({k})}n^{O(1)}$ time approximation scheme when $d$ is a constant?
\end{tcolorbox}

Note that for constant $d$, $k$-median and $k$-means already admit even better polynomial-time approximation schemes. Thus, the above questions are relevant only to $k$-center. Considering the above questions, we answer both of them in affirmative. 
\begin{theorem}[Informal]\label{thm:kcenterinformal}
For any $0 < \epsilon \leq 1$, there is an $O^*(2^{O((1/\epsilon)^{O(1)}\sqrt{k}\log k)})$-time $(1+\epsilon)$-approximation algorithm for 2-dimensional $k$-center. In general, there is an $O^*(2^{O((1/\epsilon)^{O(d)} \cdot k^{1-1/d} \cdot \log k)})$-time $(1+\epsilon)$-approximation algorithm for $k$-center in $\mathbb R^d$.
\end{theorem}

In this informal statement, the $O^*$-notation suppresses polynomial factors and a leading factor of $d^d$ in. Note that for any constant $d$, this yields a sub-exponential $O^*(2^{o(k)})$ time approximation scheme. One should compare these results with the earlier polynomial-time 1.82-factor hardness of approximation for planar $k$-center, which motivated us to design FPT algorithms. Theorem \ref{thm:kcenterinformal} is an informal version of Theorem \ref{thm:kcentercont} which appears in Section \ref{sec:kcenter}.   

In Section \ref{sec:hardness}, we prove a new hardness bound for $k$-center based on Exponential Time Hypothesis (ETH) \cite{ImpagliazzoPZ01} that almost complements our running time bound in the $2$-dimensional case. 

\begin{restatable}{theorem}{hardnesscont}
\label{hardness: cont-k-center}
There exists a constant $\alpha>1$ such that there is no $\alpha$-approximation algorithm for $2$-dimensional $k$-center in time $2^{o({k^{1/4}})}\cdot n^{O(1)}$, unless ETH fails.
\end{restatable}

Note that in the above hardness bound, the power of $k$ is only $1/4$, whereas in our upper bound it is $1/2$. 
However, in the discrete case of $k$-center (popularly known as $k$-supplier), where centers can only be chosen from a given finite set of points, we obtain a tight lower bound based on Randomized ETH (rETH) \cite{dell2014exponential}.
We note that the result of Theorem \ref{thm:kcenterinformal} also holds for $k$-supplier (Theorem \ref{thm:ksupplier}) which is described in Section \ref{sec:kcenter}. Our hardness result addresses the double-exponential dependence on $d$ in the running time by showing one cannot improve it substantially
if one hopes to have better-than-linear dependence on $k$ in the exponent.

\begin{restatable}{theorem}{hardness}
\label{hardness: k-center}
Under rETH, there are constants $\bar\epsilon>0$ and $\alpha>1$ such that there is no $\alpha$-approximation algorithm for $d$-dimensional Euclidean $k$-supplier in time $O^*(2^{2^{\bar\epsilon\cdot\delta\cdot d}\cdot k^{1-\delta}})$ for any $0< \delta< 1$.
\end{restatable}

Next, we turn our attention to another clustering problem called Non-uniform $k$-center. It is a generalization of $k$-center where it is possible to select clusters of different radii. 

\begin{definition}[Non-Uniform $k$-center (NUkC)] 
Given two sets of points $\mathcal{C}$ and $\mathcal{F}$ in a metric space $(X,\dist)$, an integer $k > 0$, $t \le k$ distinct integers  (radii) $r_1 > r_2 > \ldots > r_t > 0$ and non-negative integers $k_1,\ldots,k_t$ such that $\sum_{i=1}^t k_i=k$, the goal is to find a number (dilation) $\alpha$ and to choose $k_i$ balls centered at the points of $\mathcal{F}$ with radius $\alpha\cdot r_i$ for all $1\le i\le t$, such that the union of the chosen balls contains all the points of $\mathcal{C}$ and $\alpha$ is minimized.  
\end{definition}
 
Note that in the special case when $t=1$, the problem is basically $k$-center. NUkC was formulated by Chakrabarty~et al.~\cite{chakrabarty2020non} who also described its applications in fine-tuned clustering and vehicle routing. They showed that for any $\gamma \ge 1$, the problem in general metrics is \np-hard to approximate within a factor of $\gamma$. The same hardness holds even in the discrete (and hence in the continuous) Euclidean case \cite{bandyapadhyay2020perturbation}. Consequently, we ask the following questions. 

\begin{tcolorbox}
Does NUkC in general metrics admit a constant-factor approximation in FPT ($f(k)\cdot n^{O(1)}$) time? Does Euclidean NUkC admit an approximation scheme in FPT ($f(k)\cdot n^{O(1)}$) time?
\end{tcolorbox}

In this work, we answer both of the questions in affirmative. 

\begin{theorem}[Informal]\label{thm:nukcintro}
A $3$-approximation for NUkC can be computed in time $2^{O(k\log k)}n^2$. Moreover, if $\mathcal{C}=\mathcal{F}$, the approximation factor is 2. For Euclidean NUkC, a $(1+\epsilon)$-approximation can be computed in time $2^{O((k\log k)/\epsilon)}dn$. 
\end{theorem}

Note that the result for Euclidean NUkC is a strict generalization of the result for Euclidean $k$-center due to Badoiu and Clarkson \cite{DBLP:conf/soda/BadoiuC03}. Indeed, our algorithm is motivated by their algorithm. Theorem \ref{thm:nukcintro} is an informal version of Theorem \ref{thm:nukcgeneral} and \ref{thm:nukceuclid} which appear in Section \ref{sec:nukc}.  See Table \ref{table:1} for a summary of our results. 

%% file: kcenter.tex
We begin by presenting our faster approximation scheme for $k$-center.
In fact, we give a faster approximation scheme for the closely-related $k$-supplier
problem, and then discuss how this can be used to give a faster approximation
for $k$-center. 

\paragraph{$k$-supplier.} We are given points $\mathcal{C}$ to be clustered
and given candidate centers $\mathcal{F}$. The goal is to find $F \subseteq \mathcal{F}$
with $|F| = k$ minimizing ${\rm cost}(F) := \max_{p \in \mathcal{C}} d(p,F)$.

Throughout this section, $O()$ will suppress only absolute constants that are
independent of $d$ and $\epsilon$.

We first prove the following.
\begin{theorem}\label{thm:ksupplier}
  For any $d \geq 1$ and $1 \geq \epsilon > 0$, there is a $(1+\epsilon)$-approximation
  for instances of $k$-supplier in $\mathbb R^d$ with running time
  $O(|\mathcal C| \cdot |\mathcal F| +  2^{O((1/\epsilon)^{O(d)} \cdot k^{1-1/d} \cdot \log k + d \log d)})$.
\end{theorem}
Note, if one simply regards $d$ and $\epsilon$ as constants then the exponential term
in the running time simplifies to $2^{O(k^{1-1/d} \cdot \log k)}$. For the special
case of $\mathbb R^2$, the exponential part of the running time is $2^{O(\sqrt{k} \cdot \log k)}$.
The doubly-exponential dependence on $d$ is clearly not desirable, but we
provide evidence in the next section suggesting this may not be possible to improve.

The main idea behind our approach uses a somewhat-recent result by Bhattiprolu and Har-Peled \cite{bhattiprolu2016voronoi} on Voronoi separators. Intuitively, they show that
for any $n$ points in $\mathbb R^d$, one can insert an additional $O(n^{1-1/d})$ points
such that in the Voronoi diagram of all points (original and inserted), the two sets of original points
can be partitioned into two roughly-equal halves and these halves are separated by the Voronoi cells
of the new points.

We utilize such a Voronoi separator to help guess $O(k^{1-1/d} \cdot (1/\epsilon)^{O(d)})$ centers
of the optimum solution such that if we serve all points near these centers, the remaining problem
naturally decomposes into two roughly equal-size halves that can be treated independently (thus, bounding
the depth of recursion to be logarithmic).
Naturally, to ensure we are guessing these centers from a set of size $O(k \cdot (1/\epsilon)^{O(d)})$ rather than from a possibly-larger
set $O(|\mathcal F|)$ we also have to filter the input so points are not ``close'' to each other.

\paragraph{\bf Preliminary Step: Reducing to a feasibility check}
In $O(|\mathcal C| \cdot |\mathcal F|)$ time, we can find a solution $F \subseteq \mathcal{F}$
with ${\rm cost}(F) \leq 3 \cdot OPT$ \cite{hochbaum1986unified}. We then know $OPT \leq {\rm cost}(F) \leq 3 \cdot OPT$. In the remaining steps, we will describe a binary search algorithm that
for some value $R > 0$, will either find a $k$-supplier solution with value $(1+\epsilon) \cdot R$ or will (correctly) declare there is no solution with value $\leq R$. We then use a binary
search to find a value $R$ in the range $[{\rm cost}(F)/3, {\rm cost}(F)]$ such that using $R$
produces an infeasible solution whereas using $(1+\epsilon) \cdot R$ will return a solution with cost $(1+\epsilon)\cdot(1+\epsilon) \cdot R \leq (1+3\epsilon) \cdot R$. Scaling $\epsilon$ by a constant factor gives the desired result. The number of iterations of the binary search will be at most $\log_2 \frac{3}{\epsilon}$.

By scaling the point set, it suffices to give such an algorithm for $R = 1$. That is, from this
point forward we present an algorithm that will either find a solution with
cost at most $1+\epsilon$ or correctly declare there is no solution with value $\leq 1$.

\paragraph{Step 1: Reducing the input size}
We perform standard filtering. Initially let $\mathcal C' = \emptyset$ and $\mathcal F' = \emptyset$.
Then we process $p \in \mathcal C$ one at a time: if $d(p, \mathcal C') > \epsilon$
then we add $p$ to $\mathcal C'$. Also process $i \in \mathcal F$ one at a time: if $d(i, \mathcal F') > \epsilon$ yet $d(i,\mathcal C') \leq 1+\epsilon$, then add $i$ to $\mathcal F'$.

\begin{lemma}
  If there is a solution with cost at most 1, we have $|\mathcal C'|, |\mathcal F'| \leq k \cdot (1/\epsilon)^{O(d)}$ and the $k$-supplier instance $(\mathcal C', \mathcal F', k)$ has optimum value at most $1+2\cdot \epsilon$. Conversely, given a solution $F' \subseteq \mathcal F'$ with cost at most $\alpha$ in the new instance, its cost in the original instance is at most $\alpha + \epsilon$.
\end{lemma}
\begin{proof}
  Suppose there is a solution with cost at most 1 in the original instance. Let $F^* \subseteq F$ be the optimum centers. Notice $d(p,p') > \epsilon$ for distinct $p,p' \in \mathcal C'$. So for each $i \in F^*$, the number of points $p \in \mathcal C'$ with $d(i,p)\le 1$ is at most $(1/\epsilon)^{O(d)}$. Since ${\rm cost}(F^*) \leq 1$, all $p \in \mathcal C'$ lie within distance 1 from some $i \in F^*$. So the total number of points in $\mathcal C'$ is at most $|F^*| \cdot (1/\epsilon)^{O(d)} = k \cdot (1/\epsilon)^{O(d)}$.

  Similarly, note $d(i,i') > \epsilon$ for distinct $i,i' \in \mathcal F'$ but each $i \in \mathcal F'$ lies within distance at most $1+\epsilon$ from some $p \in \mathcal C'$.
  Therefore, $|\mathcal F'| \leq |\mathcal C'| \cdot (1/\epsilon)^{O(d)} \leq k \cdot (1/\epsilon)^{O(d)}$.

  Next, let $F^{'*}$ denote the solution obtained from $F^*$ as follows. Consider each $i \in F^*$ that is actually covering a point in $\mathcal C'$ (i.e. $d(i, \mathcal C') \leq 1)$). If $i \in \mathcal F'$, add $i$ to $F^{'*}$. Otherwise let $i' \in \mathcal F'$ satisfy $d(i,i') \leq \epsilon$. Add $i'$ to $F^{'*}$. Since each $p \in \mathcal C'$ had distance at most 1 from $F^*$, it then has distance at most $1+\epsilon$ from $F^{'*}$.

  Conversely, let $F' \subseteq \mathcal F'$ be a solution with cost $\alpha$ in the new instance. Consider any $p \in \mathcal C$. Since $d(p,p') \leq \epsilon$ for some $p' \in \mathcal C'$ (it could be $p = p'$) then $d(p, F') \leq d(p,p') + d(p',F') \leq \epsilon + \alpha$.
\end{proof}
In our algorithm, if $|\mathcal C'|$ or $|\mathcal F'|$ exceeds $k \cdot (1/\epsilon)^{O(d)}$
then we declare this instance is a {\bf no} instance and terminate. In fact,
if we terminate in the middle of this step once $|\mathcal C'|$ or $|\mathcal F'|$ becomes too large then a simple implementation runs in time $O((|\mathcal C| + |\mathcal F|) \cdot k \cdot (1/\epsilon)^{O(d)})$.

Also, if $k \leq d \cdot \log 1/\epsilon$ we simply solve the problem using brute force
within the desired time by trying all subsets of $\mathcal F'$ of size at most $k$
(there are $(1/\epsilon)^{O(d^2 \log 1/\epsilon)}$ such subsets to try).

\paragraph{Step 2: Identifying a sparse separator}
We use the following result about separators in Euclidean spaces
that mimics the planar separator theorem.
\begin{theorem}[\cite{bhattiprolu2016voronoi}]
  Let $\mathcal X \subseteq \mathbb R^d$ be a set of $n$ points
  and let $c_d := \lceil 2 \sqrt d\rceil^d + 1$.
  In expected $O(n)$ time, one can compute a set of $\mathcal Z \subseteq \mathbb R^d$
  with $|\mathcal Z| \leq O(n^{1-1/d})$ and a partition $\mathcal X_1, \mathcal X_2$
  of $\mathcal X$ such that: a) In the Voronoi diagram of $\mathcal X \cup \mathcal Z$,
  the Voronoi cells of points $p \in \mathcal X_1$ and $q \in \mathcal X_2$ do not share
  any common point and b) $|\mathcal X_1|, |\mathcal X_2| \leq n \cdot (1-1/c_d)$.
\end{theorem}
Note this is a randomized algorithm, but it is guaranteed to produce such a structure.
It is only the running time that is a random variable.

Let $\mathcal Z$ be the Voronoi separator computed for point set $\mathcal X := \mathcal C' \cup \mathcal F'$ and let $\mathcal X_1, \mathcal X_2$ be the two parts of the partition of $\mathcal X$.

\paragraph{\bf Step 3: Guessing and recursing}
Here, we guess a small set of centers in the optimum solution that lie near the Voronoi separator and recurse on both sides.

For each point $q$ in the Voronoi separator $\mathcal Z$, we guess all points in $F^{'*}$ (the optimum solution for $(\mathcal C', \mathcal F', k)$) that lie within distance at most $2 \cdot (1+\epsilon)$ of $q$. Recall $d(i,i') > \epsilon$ for distinct $i,i' \in \mathcal F'$, so there are at most $(1/\epsilon)^{O(d)}$ points to guess for this point $q$.

Let $\mathcal G$ be the union of these guesses for all $q \in \mathcal Z$, so $|\mathcal G| \leq \zeta := O(k^{1-1/d} \cdot (1/\epsilon)^{O(d)})$. Using standard estimates on binomial coefficients, the number of such guesses to enumerate is at most ${|\mathcal F'| \choose \zeta} \leq 2^{O((1/\epsilon)^d \cdot k^{1-1/d} \cdot \log k)}$.

For each such guess $\mathcal G$, we remove all points in $\mathcal C'$ that are within distance
at most $1+\epsilon$ from $\mathcal G$. Say the remaining points are $C''$. Also let $\mathcal F'' := \mathcal F' - \mathcal G$.
\begin{lemma}
  Suppose $\mathcal G\subseteq \mathcal F'$ is the proper guess. For $j = 1,2$, every point in $\mathcal C'' \cap \mathcal X_j$ is within distance at most $1+\epsilon$ from $F^{'*} \cap (\mathcal X_j - \mathcal G)$.
\end{lemma}
\begin{proof}
  Let $p \in \mathcal C'' \cap \mathcal X_j$ and suppose $i \in F^{'*}$ satisfies $d(p,i) \leq 1 + \epsilon$. Since $p \in \mathcal C''$, it must be that $i \notin \mathcal G$.
  If $i \notin \mathcal X_j$, then the straight line connecting $i$ to $p$ must touch
  the Voronoi cell for some $q \in \mathcal Z$.

  Let $t$ be the first point along the $p-i$ segment that touches the Voronoi cell for $q$.
  As this is a Voronoi diagram, we have $d(t,q) \leq d(t,i)$. So,
  \[ d(q,i) \leq d(q,t) + d(t,i) \leq 2 \cdot d(t,i) \leq 2 \cdot (1+\epsilon) \]
  since $d(t,i) \leq d(p,i) \leq 1+\epsilon$. But then $i$ would have been included in $\mathcal G$, a contradiction. So $i \in \mathcal X_j - \mathcal G$.
\end{proof}

Finally, we guess $k_j := |\mathcal F'' \cap F^{'*}|$ for both $j = 1,2$ and independently recurse starting at step 2 on each of the two instances $(\mathcal C'' \cap \mathcal X_j, \mathcal F'' \cap \mathcal X_j, k_j), j = 1,2$. For the appropriate guess $\mathcal G$, each of the two subproblems
has a feasible solution of cost at most $1+\epsilon$.

The base case is when we recurse with an
empty subproblem (i.e. $\mathcal C'' \cap \mathcal X_j = \emptyset$) in which case there are no centers to select. If a subproblem we ever recurse on has $\mathcal C' \neq \emptyset$
yet $\mathcal F' = \emptyset$, we declare this subproblem to be infeasible and return no solution.

If some guess for $\mathcal G$ and $k_1,k_2$ has both recursive calls returning
a feasible solution (i.e. of cost $1+\epsilon$), then adding the centers to $\mathcal G$
produces a feasible solution for this instance $(\mathcal C', \mathcal F', k)$ and we return it.
Otherwise, if all guesses have at least one of the two recursive calls returning no solution,
we return no solution.

\paragraph{\bf Analysis}
We argued throughout the presentation of the algorithm that
if there is a feasible solution $F^{'*}$ (of size $\leq k$ and cost $\leq 1+\epsilon$), 
it would correctly find a solution of cost $1+\epsilon$ for the original (filtered)
instance $(\mathcal C', \mathcal F', k)$ for the branches of recursion that
properly guessed $\mathcal G$ and $k_1, k_2$. One should also note $d(i,i') > \epsilon$ for distinct $i,i' \in \mathcal F'$ holds throughout all recursive calls because it holds after step 1 and we only recurse with subsets of $\mathcal F'$.

Let $\Delta := k \cdot (1/\epsilon)^{O(d)}$, a bound on the initial size of $|\mathcal F \cup \mathcal C|$.
At depth $i$ of the recursion, the size of $|\mathcal F' \cup \mathcal C'|$ is bounded by $(1-1/c_d)^i \cdot \Delta$.
The number of recursive calls spawned from a depth $i$ recursive call is at most $2 \cdot |\mathcal F'|^{(1/\epsilon)^{O(d)} \cdot |\mathcal Z|} \leq \Delta^{(1/\epsilon)^{O(d)} \cdot (1-1/c_d)^{i/2} \cdot \Delta^{1-1/d}}$ (using $1-1/d \geq 1/2$). By summing a geometric series in the exponent and recalling $c_d = d^{O(d)}$ we see the total number of recursive calls reaching depth $i$ is at most
$2^i \cdot \Delta^{O((1/\epsilon)^{O(d)}\cdot d^{O(d)} \cdot \Delta^{1-1/d})}$.
Noting the depth of recursion is at most $\log_{1/(1-1/c_d)}\Delta = O(c_d \cdot \log \Delta)$
and time taken between subsequent recursive calls is polynomial in $|\mathcal F|$, we have that the total running time of steps 1 through 3 is $2^{O((1/\epsilon)^{O(d)} \cdot k^{1-1/d} \cdot \log k + d \log d)}$

\subsection{From $k$-Supplier to $k$-Center}
We now turn to classic $k$-center in Euclidean spaces. Here, we are given points $\mathcal C$
and a value $k$. The goal is to find $k$ points $F \subseteq \mathbb R^d$ minimizing
$\max_{p \in \mathcal C}$.
\begin{theorem}\label{thm:kcentercont}
  For any $d \geq 1$ and $\epsilon > 0$, there is a $(1+\epsilon)$-approximation
  for instances of $k$-center in $\mathbb R^d$ with running time
  $O(|\mathcal C| \cdot k \cdot (1/\epsilon)^{O(d)} + 2^{O((1/\epsilon)^{O(d)} \cdot k^{1-1/d} \cdot \log k + d \log d)})$.
\end{theorem}
\begin{proof}
  Just like in the $k$-supplier algorithm, it suffices to give an algorithm
  with running time $2^{O((1/\epsilon)^{O(d)} \cdot k^{1-1/d} \cdot \log k + d \log d)}$
  that either finds a solution with cost $1+\epsilon$ or determines there is no solution with
  cost at most $1$.

  Begin by forming $\mathcal C'$ as in step 1 of the $k$-supplier algorithm:
  so $d(p,p') > \epsilon$ for distinct $p,p' \in \mathcal C$ but $d(p, \mathcal C') \leq \epsilon$ for each $p \in \mathcal C$. Just like in step 1 of the $k$-supplier algorithm, if there is a solution with cost $\leq 1$, then $|\mathcal C'| \leq k \cdot (1/\epsilon)^{O(d)}$.
  If $|\mathcal C'|$ is larger than this, declare there is no solution.

  Finally, for each $p \in \mathcal C'$ we let $\mathcal F'_p$ be any $\epsilon$-net
  of the ball of radius $1+\epsilon$ centered at $\mathcal C'$. Then $|\mathcal F'_p| \leq (1/\epsilon)^{O(d)}$ yet $d(q,\mathcal F'_p) \leq \epsilon$ for any $q$ with $d(p,q) \leq 1+\epsilon$.
  Just like in step 1 of the $k$-supplier algorithm, the optimum solution to the $k$-center
  instance $\mathcal C'$ has a feasible solution with cost at most $1+\epsilon$
  by restricting the possible locations to $\mathcal F'=\cup_{p \in \mathcal C'} \mathcal F'_p$. Run the $k$-supplier algorithm
  with these $\mathcal C'$ and $\mathcal F'$.
\end{proof}

%% file: hardness.tex
In this section, we prove our hardness results. First, we describe the result for 
the Euclidean $k$-supplier problem. 

Recall that $O^*$-notation suppresses factors polynomial in the input size. We use randomized Exponential Time Hypothesis (rETH) as our complexity theoretic assumption. Given a 3-SAT instance, we denote by $n$ and $m$ the number of variables and clauses in the instance, respectively. Then, rETH states that there is a constant $c>0$ such that there is no randomized algorithm that decides 3-SAT in time $O^*(2^{c\cdot n})$ with (two-sided) error probability at most $\frac{1}{3}$ \cite{dell2014exponential}. Using the sparsification lemma \cite{jansen2010kernelization} one can show the following, see Exercises 14.1 of \cite{cygan2015parameterized}.

\begin{theorem}\label{hardness: rETH}
Under rETH, there exists a constant $c>0$ such that there is no randomized algorithm that decides 3-SAT in time $O^*(2^{c\cdot (n+m)})$ with (two-sided) error probability at most $\frac{1}{3}$ where $n$ and $m$ are the numbers of variables and clauses in the 3-SAT instance, respectively.
\end{theorem}

There is a known reduction from 3-SAT to Vertex Cover such that the number of vertices in the Vertex Cover instance is linear in the number of variables and clauses of 3-SAT instance, see \cite{cygan2015parameterized} for this reduction. Using such reduction and Theorem \ref{hardness: rETH} we have the following hardness result on Vertex Cover.

\begin{theorem}[Hardness of Vertex Cover]\label{hardness: vertex cover}
Under rETH, there is a constant $c>0$ such that there is no randomized algorithm that decides Vertex Cover with $n$ vertices in time $O^*(2^{c\cdot n})$ with (two-sided) error probability at most $\frac{1}{3}$.
\end{theorem}

The last ingredient we need to show our hardness result for Euclidean $k$-supplier problem is the famous Johnson-Lindenstrauss dimentionality reduction \cite{johnson1984extensions} or JL lemma for short.

\begin{theorem}[JL lemma, reformulated]\label{hardness: JL}
Consider $n$ points $x_1,...,x_n$ in $\mathbb{R}^{\ell}$. There exists a linear map $A:\mathbb{R}^\ell\to\mathbb{R}^{c_{JL}\cdot\log n}$ where $c_{JL}>0$ is a constant such that for all $1\leq i,j\leq n$ we have
\begin{equation}\label{hardness: JL ineq}
0.9\cdot\norm{x_i-x_j}\leq \norm{A(x_i)-A(x_j)}\leq 1.1\cdot\norm{x_i-x_j},
\end{equation}
where the norms are $l_2$-norm. Furthermore, $A$ can be computed in polynomial time by a randomized algorithm such that (\ref{hardness: JL ineq}) holds for all pairs of given points with probability at least $\frac{2}{3}$.
\end{theorem}

Now we are ready to state and prove our hardness result.
\hardness*
\begin{proof}
We show a gap-introducing reduction from Vertex Cover to Euclidean $k$-supplier problem. 
Consider a Vertex Cover instance $(G=(V,E),k)$ where $V=\{1,2,...,n\}$. We construct an instance of Euclidean $k$-supplier as follows: define the client set as $\mathcal{C}:=\{e_i+e_j:~\forall~(i,j)\in E\}$ and the set of facilities as $\mathcal{F}:=\{e_i:~\forall 1\leq i\leq n\}$ where $e_i$ is the standard unit vector in $\mathbb{R}^n$, i.e., $e_i$ has $1$ in $i$-th coordinate and $0$ elsewhere. 

Suppose $G$ has a vertex cover $S=\{i_1,...,i_k\}$. We claim the union of balls of radius $1$ around each center $e_{i_j}$ for all $1\leq j\leq k$ covers all the clients. The reason is that for each client $e_r+e_s$, either $r$ or $s$ is in $S$. W.l.o.g., we assume $r\in S$, and therefore we open $e_r$ as a center with radius $1$ and the claim follows by noting that $\norm{e_r-(e_r+e_s)}=1$.

Next, we prove the other direction. Suppose $G$ does not have a vertex cover of size $k$. We show that an optimal solution for $k$-supplier instance is at least $\sqrt{3}$. Let $F^*$ be an optimal set of centers. We have the following fact.
\begin{claim}
There is a client $e_r+e_s$ that is assigned to $e_{i^*}\in F^*$ where $i\neq r,s$.
\end{claim}
Note $\norm{e_{i^*}-(e_r+e_s)}=\sqrt{3}$ and therefore the optimal solution for the $k$-supplier instance is at least $\sqrt{3}$. So it remains to prove the claim.
\begin{proof}(of claim)
Suppose not. Then, every client $e_r+e_s$ is assigned to either center $e_r$ or $e_s$. It is easy to see if we define $S$ as the set of all vertices $i$ such that $e_i\in F^*$, then $S$ is a feasible vertex cover for $G$, a contradiction.
\end{proof}

We can use the JL lemma (Theorem \ref{hardness: JL}) to further reduce the above $k$-supplier instance in $\mathbb{R}^n$ to a $k$-supplier instance in dimension $c_{JL}\cdot\log n$ (recall that $c_{JL}$ is a constant in the JL lemma). Then, with probability at least $\frac{2}{3}$, a YES-instance of Vertex Cover is mapped to an instance of $k$-supplier in dimension $c_{JL}\cdot\log n$ with optimal solution at most $1.1$ and a NO-instance of Vertex Cover is mapped to an instance of $k$-supplier in dimension $c_{JL}\cdot\log n$ with optimal solution at least $\sqrt{3}\cdot 0.9$. 

Finally, set $\alpha:= \frac{\sqrt{3}\cdot 0.9}{1.1}> 1.4$ and $\bar \epsilon:=\frac{1}{2\cdot c_{JL}}$ (in the statement of Theorem). We finish the proof by way of contradiction. Suppose there is a $1.5$-approximation algorithm for $k$-supplier in $\mathbb{R}^d$ that runs in $O^*(2^{2^{\bar\epsilon\cdot\delta\cdot d}\cdot k^{1-\delta}})$. Therefore, we can decide Vertex Cover in time $O^*(2^{2^{\bar\epsilon\cdot\delta\cdot d}\cdot k^{1-\delta}})$ with error probability at most $\frac{1}{3}$ where $d=c_{JL}\cdot\log n$. Plugging their values of $d,\alpha,\bar \epsilon$ and noting that $k\leq n$, we have a randomized algorithm that decides Vertex Cover in time $O^*(2^{n^{1-\frac{\delta}{2}}})$ for some $0<\delta<1$ and this contradicts the hardness result for Vertex Cover (Theorem \ref{hardness: vertex cover}).

\end{proof}

\subsection{Hardness of $k$-center}

Our hardness reduction is from 3-Planar Vertex Cover which is a restricted version of Vertex Cover on planar graphs of maximum degree 3. We need the following complexity result based on Exponential Time Hypothesis (ETH) \cite{ImpagliazzoPZ01}. 

\begin{proposition}\label{prop:planarVC}
There is no $2^{o(\sqrt{n})}$ time algorithm for 3-Planar Vertex Cover, unless ETH fails, where $n$ is the number of vertices.  
\end{proposition}

\begin{proof}
The proposition for Planar Vertex Cover essentially follows by combining two reductions due to \cite{lichtenstein1982planar}: (i) 3-SAT to Planar 3-SAT and (ii) Planar 3-SAT to Planar Vertex Cover. For a more formal exposition, see Theorem 14.9 \cite{CyganFKLMPPS15} which states that there is no $2^{o(\sqrt{n})}$ time algorithm for Planar Vertex Cover, unless ETH fails. Our proposition follows from the fact that the second reduction ensures that the constructed graph has maximum degree 3. 
\end{proof}

\hardnesscont*

\begin{proof}
For proving the theorem, we use the gap-preserving reduction of Feder and Greene \cite{feder1988optimal} from 3-Planar Vertex Cover to Planar $k$-center. They start with any embedding of the planar graph where each edge is replaced by an odd length path. Let $L$ be the sum of the lengths of the edges in the embedded graph. Their reduction ensures that if there is a vertex cover of size $k_1$, then there is a $k$-center solution of radius 1, for a suitable $k\le k_1+L$, and if there is no vertex cover of size $k_1$, then there is no $k$-center solution of radius at most 1.82. By using the planar embedding scheme in \cite{shiloach1976linear}, which ensures that $L=O(n^2)$, it follows that in the constructed instances of $k$-center, $k=O(n^2)$. Thus, using an $\alpha$-approximation algorithm for Planar $k$-center with $\alpha\le 1.82$ that runs in $2^{o(k^{1/4})}\cdot n^{O(1)}$ time, we can solve 3-Planar Vertex Cover exactly in $2^{o(k^{1/4})}\cdot n^{O(1)}=2^{o(\sqrt{n})}\cdot n^{O(1)}$ time. This contradicts Proposition \ref{prop:planarVC}, and hence our claim must be true. 
\end{proof}

%% file: nukc.tex
By standard scaling argument \cite{chakrabarty2020non}, we can assume that the optimal dilation is 1. Let OPT be any optimal set of balls. We denote a ball with center $c$ and radius $r$ by $B(c,r)$. Consider any set of points $S$. A ball $B$ is said to cover $S$ if the points of $S$ are contained in $B$. A set of balls $\B$ is said to cover $S$ if the points of $S$ are contained in the union of the balls in $\B$. 
\subsection{The General Metric Case}
Note that $\mathcal{C}$ is the set of points that we need to cover and $\mathcal{F}$ is the set of centers. Let $n=| \mathcal{C}\cup \mathcal{F}|$. In this case, we give a simple 3-approximation that runs in $2^{O(k)}n^2$ time. 

\paragraph{The Algorithm.} Let $\sol$ be the solution set of balls which is initialized to $\emptyset$. Also let $\mathcal{C}'\subseteq \mathcal{C}$ be the set of points left to cover, i.e., the set of points which are not covered by $\sol$. If $\mathcal{C}'$ is empty or $|\sol|=k$, then terminate. Otherwise, proceed as follows. Consider a point $p\in \mathcal{C}'$. Suppose $r_i$ be the radius of a ball in OPT that covers $p$ and assume that we know this index $i$ where $1\le i\le t$. Consider the point $c\in \mathcal{F}$ closest to $p$. Add the ball $B(c,3r_i)$ to $\sol$. Repeat the above steps. 

The assumption that for a point we know the optimal index $i$ can be removed in a trivial manner by trying all $t\le k$ possible choices. The algorithm runs for at most $k$ iterations, and thus we make the assumption for at most $k$ points. The total number of choices to consider is thus $k^{O(k)}$. Now, in each iteration we need to find $c$ and the ball $B(c,3r_i)$, and remove the points in $B(c,3r_i)$ from consideration. We can precompute and store the closest point in $\mathcal{F}$ for each point in $\mathcal{C}$. This helps us implement each iteration in $O(n^2)$ time. Hence, the algorithm runs in total $2^{O(k\log k)}n^2$ time. Next, we prove the correctness. 

\begin{lemma}
Consider the case when the algorithm makes all correct choices. For every iteration $j$ of the algorithm except the last one, where $1\le j\le k$, there is a ball $B_l^*$ in OPT, such that in the beginning of the iteration, $B_l^*$ is not covered by $\sol$, but in the end of the iteration, $B_l^*$ is covered by $\sol$ once we add the ball $B(c,3r_i)$ to $\sol$. Moreover, the radius of $B_l^*$ is $r_i$. 
\end{lemma}

\begin{proof}
Consider any iteration $j$ which is not the last one. Thus, $\mathcal{C}'$ is non-empty and we pick a point $p\in \mathcal{C}'$. As $p$ is not yet covered by $\sol$, there is a ball $B_l^*$ in OPT containing $p$ such that $B_l^*$ is not covered by $\sol$. We correctly find the radius $r_i$ of this ball by our assumption. Now, the center $c^*$ of $B_l^*$ must be at most at a distance $r_i$ from $p$, as $p\in B_l^*$. Thus, the closest point $c$ computed for $p$ must be at a distance $r_i$ from $p$. Now, for any point $p'\in B_l^*$, $\dist(c,p')\le \dist(c,p)+\dist(p,p')\le r_i+\dist(p,c^*)+\dist(c^*,p')\le r_i+r_i+r_i=3r_i$. (The first and the second inequalities are due to triangle inequality.) Thus, any point in $B_l^*$ is contained in $B(c,3r_i)$. As this ball is added to $\sol$ in the end of the iteration, $B_l^*$ is covered by $\sol$. 
\end{proof}

The above lemma shows that in every iteration of the algorithm except the last one, a new ball in OPT is being covered by a ball of radius 3 times the radius of the optimal ball. Thus, $\sol$ must cover all the points in $\mathcal{C}$ after at most $k$ iterations. We note that if $\mathcal{C}=\mathcal{F}$, then an algorithm similar to the above indeed gives a 2-approximation, as in that case, one can pick the ball $B(p,2r_i)$ in every iteration. The analysis is very similar. We obtain the following theorem. 

\begin{theorem}\label{thm:nukcgeneral}
A $3$-approximation for NUkC can be computed in time $2^{O(k\log k)}n^2$. Moreover, if $\mathcal{C}=\mathcal{F}$, the approximation factor is 2. 
\end{theorem}

\subsection{The Euclidean Case}
In this case, $\mathcal{C}$ is a subset of $n$ points of $\R^d$ and $\mathcal{F}=\R^d$. A ball $B=B(c,r)$ with $c\in \R^d$ is called the Minimum Enclosing Ball (MEB) of $S$ if $B$ covers $S$ and there is no ball with center in $\R^d$ and radius strictly less than $r$ that covers $S$. We denote the center and radius of the MEB of $S$ by $c_{B(S)}$ and $r_{B(S)}$, respectively. 

Our algorithm is basically an extension of the $k$-center algorithm due to Badoiu~et al.~
\cite{DBLP:conf/stoc/BadoiuHI02}. 
In $k$-center, the goal is to select $k$ balls of equal radius that cover all the input points. However, in our case, we need to find $k_i$ balls of radius $r_i$ for each $1\le i\le t$ to cover the input points. Our main contribution is to handle non-uniform radii. 

\paragraph{The Algorithm.} Consider any arbitrary ordering of the $k$ balls $B_1^*=B(c_1^*,r_1^*),\ldots,B_k^*=B(c_k^*,r_k^*)$ in OPT. For each index $1\le i\le k$, we maintain a set $S_i\subseteq \mathcal{C}\cap B_i^*$ of points which is initialized to $\emptyset$. To start with, we pick any arbitrary point $q\in \mathcal{C}$ and add it to the set $S_i$ such that $q\in B_i^*$. For the time being, assume that we know this index $i$ for $q$. Next, we apply the following procedure over $k\cdot \lceil 2/\epsilon\rceil$ iterations. In the beginning of every iteration, we check whether all the points in $\mathcal{C}$ are covered by the union of the balls $\cup_{i=1}^k B(c_{B(S_i)},(1+\epsilon)r_i^*)$. If yes, then we terminate. Otherwise, we proceed as follows. For each point $p\in \mathcal{C}$, let $\delta(p)=\min_{i=1}^k \lVert c_{B(S_i)}-p\rVert$. Pick a point $p'\in \mathcal{C}$ among the points not in $\cup_{i=1}^k B(c_{B(S_i)},(1+\epsilon)r_i^*)$ such that $\delta(p')$ is maximized, i.e., $p'$ is a point which is farthest away from the centers of the MEBs among the \say{uncovered} points. Again suppose we know the correct optimal ball $B_i^*$ that contains $p'$. Add $p'$ to $S_i$. 

The assumption that for a point we know its optimal ball can be removed in a trivial manner by trying all $k$ possible choices. As we make the assumption for $O(k/\epsilon)$ points, the total number of choices to consider is $k^{O(k/\epsilon)}$. In every iteration, computation of the MEBs and $p'$ can be done in $O(ndk)$ time. Thus, the algorithm runs in total $2^{O((k\log k)/\epsilon)}dn$ time. 

It is possible to save some computation in the above algorithm by removing all points from consideration which are already covered and do not appear in a set $S_i$. However, the asymptotic complexity remains the same.  

Next, we prove the correctness of the algorithm. In particular, we show that when the algorithm terminates, all points of $\mathcal{C}$ are covered by $\cup_{i=1}^k B(c_{B(S_i)},(1+\epsilon)r_i^*)$. We need the following lemma proved in \cite{DBLP:conf/soda/GoelIV01}. 

\begin{lemma}\label{lem:distancelb}
For any set of points $T\subset \R^d$ and any point $z$ at distance $K$ from $c_{B(T)}$, there is a point $z'\in T$ at distance at least $\sqrt{r_{B(T)}^2+K^2}$ from $z$.  
\end{lemma}

First, we prove the following lemma. 

\begin{lemma}\label{lem:boundoncoreset}
Consider any set $S_i$ where $1\le i\le k$ and suppose the algorithm added $\lceil 2/\epsilon\rceil$ points of $\mathcal{C}\cap B_i^*$ to $S_i$. Then for any point $p\in \mathcal{C}\cap B_i^*$, 
$\lVert c_{B(S_i)}-p\rVert\le (1+\epsilon)r_i^*$. 
\end{lemma}

\begin{proof}
Consider the first 
$\tau=\lceil 2/\epsilon\rceil$ iterations when the algorithm adds points of $\mathcal{C}\cap B_i^*$ to the set $S_i$. Let $S_{i,j}$ be the set $S_i$ at the $j$-th such iteration.  Thus, $S_{i,0}$ is a singleton set, and $S_{i,j+1}=S_{i,j}\cup \{p'\}$ for some $p'\in \mathcal{C}\cap B_i^*$ added to $S_{i,j}$.  
Also let $\hat{r}_j=r_{B(S_{i,j})}$, $\hat{R}=(1+\epsilon)r_i^*$, $\lambda_{j}=\hat{r}_j/\hat{R}$, and $K_j=\lVert c_{B(S_{i,j+1})}-c_{B(S_{i,j})}\rVert$. We have the following claim whose proof is similar to the proof of a theorem (Theorem 2.1) due to Badoiu and Clarkson \cite{DBLP:conf/soda/BadoiuC03}. 

\begin{claim}
$\lambda_{j}\ge 1-\frac{1}{1+j/2}$. 
\end{claim}

\begin{proof}
Note that 
the point $p'$ that the algorithm adds to $S_{i,j}$ has the property that it is currently not in $\cup_{l=1}^k \{B(c_{B(S_l)},(1+\epsilon)r_l^*)\}$ and, in particular not in $B(c_{B(S_{i,j})},\hat{R})$. Moreover, $p'$ is farthest away from $c_{B(S_{i,j})}$ among the points in $\mathcal{C}\cap B_i^*$ not covered by $\cup_{l=1}^k \{B(c_{B(S_l)},(1+\epsilon)r_l^*)\}$. 
Thus, $\lVert c_{B(S_{i,j})}-p'\rVert >  \hat{R}$. Also, $\lVert c_{B(S_{i,j})}-p'\rVert\le \lVert c_{B(S_{i,j})}-c_{B(S_{i,j+1})}\rVert + \lVert c_{B(S_{i,j+1})}-p'\rVert\le K_j+\hat{r}_{j+1}$. Thus, $\hat{r}_{j+1} > \hat{R}-K_j$. 

By Lemma \ref{lem:distancelb}, with $S_{i,j}=T$ and $c_{B(S_{i,j+1})}=z$, there is a point of $S_{i,j}$ at least $\sqrt{\hat{r}_j^2+K_j^2}$ from $c_{B(S_{i,j+1})}$. Thus, $\hat{r}_{j+1} \ge \sqrt{\hat{r}_j^2+K_j^2}$. It follows that, \[\lambda_{j+1}\hat{R}=\hat{r}_{j+1}\ge \max\{\hat{R}-K_j,\sqrt{\lambda_{j}^2{\hat{R}}^2+K_j^2}\}\]
The lower bound on $\lambda_{j+1}$ is smallest when the above two quantities are equal, i.e., when $K_j=\frac{(1-\lambda_j^2)\hat{R}}{2}$. Thus, $\lambda_{j+1}\ge (1+\lambda_{j}^2)/2$. 

This recurrence solves to $\lambda_{j}\ge 1-\frac{1}{1+j/2}$ using the fact that $\lambda_0=0$. 
\end{proof}

Now, we know that $S_{i,j} \subseteq \mathcal{C}\cap B_i^*$ and so the radius of the MEB of $S_{i,j}$ can be at most $r_i^*$. Hence, $\lambda_{j}$ can be at most $1/(1+\epsilon)$ by definition. By the above claim, this maximum value is achieved when $j\ge 2/\epsilon$, and thus $j< 1+\lceil 2/\epsilon\rceil$. It follows that after $\lceil 2/\epsilon\rceil$ points of $\mathcal{C}\cap B_i^*$ are added to $S_i$, $(1+\epsilon)$-expansion of $B(S_i)$ contains all the points of $\mathcal{C}\cap B_i^*$. As the radius of $B(S_i)$ is at most $r_i^*$, for any point $p\in B_i^*$, $\lVert c_{B(S_i)}-p\rVert\le (1+\epsilon)\cdot r_{B(S_{i})}\le (1+\epsilon)r_i^*$. 
\end{proof}

Next, we finish the correctness proof. 

\begin{lemma}
When the algorithm terminates, all points of $\mathcal{C}$ are covered by $\cup_{i=1}^k B(c_{B(S_i)},(1+\epsilon)r_i^*)$.
\end{lemma}

\begin{proof}
Suppose the statement is not true, i.e., there is a point $p$ which is not covered by $\cup_{i=1}^k B(c_{B(S_i)},(1+\epsilon)r_i^*)$. This means the algorithm ran for all $k\cdot\lceil 2/\epsilon\rceil$ iterations, i.e., the total size of the sets $\{S_i\}$ is $k\cdot\lceil 2/\epsilon\rceil$. Let $j$ be the index such that $p\in B_j^*$. Now, consider the case when the algorithm makes all correct choices. Note that there must be at least one such case, as we try all possible $k$ choices in every step. By Lemma \ref{lem:boundoncoreset} and the way the algorithm adds points to the sets $\{S_i\}$, the size of each $S_i$ can be at most $\lceil 2/\epsilon\rceil$. But, this implies that the size of $S_j$ is at least $\lceil 2/\epsilon\rceil$, as the total size of the sets $\{S_i\}$ is $k\cdot\lceil 2/\epsilon\rceil$. It follows by Lemma \ref{lem:boundoncoreset}, $\lVert c_{B(S_j)}-p\rVert\le (1+\epsilon)r_j^*$, which is a contradiction. 
\end{proof}

From the above discussion, we obtain the following. 

\begin{theorem}\label{thm:nukceuclid}
A $(1+\epsilon)$-approximation for Euclidean NUkC can be computed in time $2^{O((k\log k)/\epsilon)}dn$. 
\end{theorem}